\documentclass[12pt,a4paper]{article}
\usepackage{amsmath,amssymb,amsthm} 
\numberwithin{equation}{section}
\usepackage[dvipdfmx]{graphicx,color}
\newtheorem{lemma}{Lemma}
\newtheorem{theorem}{Theorem}

\theoremstyle{remark}
\newtheorem{remark}{Remark}
\newcommand{\beq}{\begin{equation}}
\newcommand{\eeq}{\end{equation}}
\newcommand{\beqnn}{\begin{equation*}}
\newcommand{\eeqnn}{\end{equation*}}

\newcommand{\tp}[1]{\,{\vphantom{#1}}^\mathrm{t}\!\,#1}
\newcommand{\CC}{\mathbb{C}}
\newcommand{\PP}{\mathbb{P}}
\newcommand{\RR}{\mathbb{R}}
\newcommand{\ZZ}{\mathbb{Z}}
\newcommand{\calN}{\mathcal{N}}
\newcommand{\calP}{\mathcal{P}}
\newcommand{\bst}{\boldsymbol{t}}
\newcommand{\bsx}{\boldsymbol{x}}
\newcommand{\bsT}{\boldsymbol{T}}
\newcommand{\bszero}{\boldsymbol{0}}
\newcommand{\frD}{\mathrm{4D}}

\begin{document}

\title{4D limit of melting crystal model\\
and its integrable structure}
\author{Kanehisa Takasaki\thanks{E-mail: takasaki@math.kindai.ac.jp}\\
{\normalsize Department of Mathematics, Kindai University}\\ 
{\normalsize 3-4-1 Kowakae, Higashi-Osaka, Osaka 577-8502, Japan}}
\date{}
\maketitle 

\begin{abstract}
This paper addresses the problems of quantum spectral curves and 4D limit 
for the melting crystal model of 5D SUSY $U(1)$ Yang-Mills theory on 
$\mathbb{R}^4\times S^1$.  The partition function $Z(\boldsymbol{t})$ 
deformed by an infinite number of external potentials is a tau function 
of the KP hierarchy with respect to the coupling constants $\boldsymbol{t} 
= (t_1,t_2,\ldots)$.  A single-variate specialization $Z(x)$ of 
$Z(\boldsymbol{t})$ satisfies a $q$-difference equation representing 
the quantum spectral curve of the melting crystal model. In the limit 
as the radius $R$ of $S^1$ in $\mathbb{R}^4\times S^1$ tends to $0$, 
it turns into a difference equation for a 4D counterpart $Z_{\mathrm{4D}}(X)$ 
of $Z(x)$. This difference equation reproduces the quantum spectral curve 
of Gromov-Witten theory of $\mathbb{CP}^1$.  $Z_{\mathrm{4D}}(X)$ is obtained 
from $Z(x)$ by letting $R \to 0$ under an $R$-dependent transformation 
$x = x(X,R)$ of $x$ to $X$.  A similar prescription of 4D limit can be 
formulated for $Z(\boldsymbol{t})$ with an $R$-dependent transformation 
$\boldsymbol{t} = \boldsymbol{t}(\boldsymbol{T},R)$ of $\boldsymbol{t}$ 
to $\boldsymbol{T} = (T_1,T_2,\ldots)$.  This yields a 4D counterpart 
$Z_{\mathrm{4D}}(\boldsymbol{T})$ of $Z(\boldsymbol{t})$.  
$Z_{\mathrm{4D}}(\boldsymbol{T})$ agrees with a generating function of 
all-genus Gromov-Witten invariants of $\mathbb{CP}^1$.  
Fay-type bilinear equations for $Z_{\mathrm{4D}}(\boldsymbol{T})$ 
can be derived from similar equations satisfied by $Z(\boldsymbol{t})$.  
The bilinear equations imply that $Z_{\mathrm{4D}}(\boldsymbol{T})$, too, 
is a tau function of the KP hierarchy.  These results are further 
extended to deformations $Z(\bst,s)$ and $Z_{\frD}(\bsT,s)$ 
by a discrete variable $s \in \mathbb{Z}$, which are shown to be 
tau functions of the 1D Toda hierarchy.  
\end{abstract}

\begin{flushleft}
\end{flushleft}

\newpage

\section{Introduction}

The melting crystal model \cite{MNTT04} is a statistical model 
of 5D SUSY Yang-Mills theory on $\RR^4\times S^1$ \cite{Nekrasov96} 
in the self-dual background \cite{Nekrasov02,NO03}.  
The partition function is a sum over all possible shapes 
(represented by plane partitions) of 3D Young diagrams. 
The name of the model originates in the physical interpretation 
of the complement of a 3D Young diagram in the positive octant 
of $\RR^3$ as a melting crystal corner. 
By the method of diagonal slicing \cite{ORV03}, 
the partition function can be converted to a sum 
over ordinary partitions.  This sum reproduces 
the Nekrasov partition function of instantons 
in 5D SUSY Yang-Mills theory. 

In the previous work \cite{NT07}, we studied the simplest case 
that amounts to $U(1)$ gauge theory.  The main subject 
was an integrable structure of the partition function 
deformed by an infinite number of external potentials.  
The deformed partition function $Z(\bst,s)$ 
depends on the coupling constants $\bst = (t_1,t_2,\ldots)$ 
of those potentials and a discrete variable $s \in \ZZ$. 
We proved, with the aid of symmetries of a quantum torus algebra, 
that $Z(\bst,s)$ is essentially a tau function of 
the 1D Toda hierarchy \cite{UT84}.  This result has been extended 
to some other types of melting crystal models \cite{Takasaki13,Takasaki14}. 

An open problem raised therein is to find an appropriate prescription 
for the 4D limit as the radius $R$ of $S^1$ in $\RR^4\times S^1$ 
tends to $0$.  The melting crystal model of $U(1)$ gauge theory 
has two parameters $q,Q$.  By setting these parameters 
in a particular $R$-dependent form and letting $R \to 0$, 
the undeformed partition function $Z = Z(\bszero,0)$ 
converges to the 4D Nekrasov function $Z_{\frD}$ \cite{Nekrasov02,NO03}.  
It is not so straightforward to achieve the 4D limit 
of the deformed partition function $Z(\bst,s)$. 
In a naive prescription \cite{NT07}, 
all coupling constants other than $t_1$ decouple 
from $Z(\bst,s)$ in the limit as $R \to 0$.  
On the other hand, a deformation $Z_{\frD}(\bst,s)$ of $Z_{\frD}$ 
by an infinite number of external potentials 
is proposed in the literature \cite{MN06}. 
What we need is an $R$-dependent transformation 
$\bst = \bst(\bsT,R)$ of $\bst$ to a new set of 
coupling constants $\bsT = (T_1,T_2,\ldots)$ 
such that $Z(\bst(\bsT,R),s)$ converges to $Z_{\frD}(\bsT,s)$ 
as $R \to 0$.  

As a warm-up for tackling this problem, we consider 
the so called quantum spectral curves. This is inspired 
by the work of Dunin-Barkowski et al. \cite{DBMNPS13} 
on Gromov-Witten theory of $\CC\PP^1$.  
They derived a quantum spectral curve of $\CC\PP^1$ 
in the perspective of topological recursion \cite{DBOSS12,DBSS12,NS11}.   
Since the deformed 4D Nekrasov function $Z_{\frD}(\bsT,s)$ 
of $U(1)$ gauge theory coincides with a generating function 
of all genus Gromov-Witten invariants of $\CC\PP^1$ \cite{LMN03,OP02}, 
it will be natural to reconsider this issue 
from the point of view of the melting crystal model. 

Recently, we proposed a new approach to quantum mirror curves 
in topological string theory \cite{TN16}. 
This approach is based on the notions of Kac--Schwarz operators 
\cite{KS91,Schwarz91} and generating operators 
\cite{Alexandrov1404,ALS1512} in the KP hierarchy \cite{SS83,SW85}. 
$Z(\bst,s)$ may be thought of as a set of KP tau functions 
labelled by $s \in \ZZ$.  In particular, 
$Z(\bst) = Z(\bst,0)$ resembles the tau functions 
in topological string theory.  
Our method developed for topological string theory 
can be applied to $Z(\bst)$ to derive a quantum spectral curve.  
This quantum curve is represented by a $q$-difference equation 
for a single-variate specialization $Z(x)$ of $Z(\bst)$.  

We show that this $q$-difference equation 
turns into the quantum spectral curve of Dunin-Barkowski et al. 
as $R \to 0$.  To this end, we choose 
an $R$-dependent transformation $x = x(X,R)$ of $x$ 
to a new variable $X$.  $Z(x(X,R))$ thereby converges 
to a function $Z_{\frD}(X)$ as $R \to 0$.  
$Z_{\frD}(X)$ is exactly the function considered 
by Dunin-Barkowski et al. and shown to satisfy 
a difference equation that represents 
their quantum spectral curve.  Although being rather simple, 
the transformation $x = x(X,R)$ is indispensable 
to achieve this limit. 

This result can be further extended to the multi-variate 
partition function $Z(\bst)$.  We construct an $R$-dependent 
transformation $\bst = \bst(\bsT,R)$ of the coupling constants $\bst$, 
and show that $Z(\bst(\bsT,R))$ does converge to the correct 
4D counterpart $Z_{\frD}(\bsT) = Z_{\frD}(\bsT,0)$ as $R \to 0$. 
The same transformation of the coupling constants can be used 
to derive $Z_{\frD}(\bsT,s)$ from $Z(\bst,s)$ as well. 
The construction of this transformation is far more complicated 
than the transformation $x = x(X,R)$ of the variable $x$.  
The complexity stems from a structural difference 
in the external potentials of $Z(\bst)$ and $Z_{\frD}(\bsT)$, 
hence being unavoidable. We believe that our choice is 
nevertheless the most natural  among all possible prescriptions. 

As a byproduct of this prescription of 4D limit, 
we can show that $Z_{\frD}(\bsT)$ satisfies a set of 
Fay-type bilinear equations.  These bilinear equations 
characterize all tau functions of the KP hierarchy 
\cite{AvM92,SS83,TT95}. This implies that $Z_{\frD}(\bsT)$, 
too, is a tau function of the KP hierarchy.  
Moreover, by a similar characterization of tau functions 
of the Toda hierarchy \cite{Takasaki07,Teo06}, 
we can deduce that $Z_{\frD}(\bsT,s)$ is a tau function 
of the 1D Toda hierarchy.  Actually, these functions 
are known as generating functions of all-genus 
Gromov-Witten invariants of $\CC\PP^1$ \cite{LMN03,OP02}.  
The Toda conjecture on Gromov-Witten theory of $\CC\PP^1$ 
\cite{DZ04,Getzler01,Milanov06,Pandharipande02} 
can be thus explained in a different perspective.  

This paper is organized as follows.  
Section 2 is a brief review of the melting crystal model. 
Combinatorial and fermionic expressions of the deformed 
partition function $Z(\bst)$ are introduced.  
The fermionic expression is further converted to a form 
that fits into the method of our work on quantum mirror curves 
of topological string theory \cite{TN16}. 
Section 3 presents the quantum spectral curve 
of the melting crystal model.  
The single-variate specialization $Z(x)$ of $Z(\bst)$ 
is introduced, and shown to satisfy a $q$-difference equation 
representing the quantum spectral curve.  
The computations are mostly parallel to the case of 
topological string theory. 
Section 4 deals with the issue of 4D limit.  
The $R$-dependent transformations $x = x(X,R)$ 
and $\bst = \bst(\bsT,R)$ are introduced, 
and $Z(x(X,R))$ and $Z(\bst(\bsT,R))$ are shown 
to converge as $R \to 0$.  The functions $Z_{\frD}(X)$ 
and $Z_{\frD}(\bsT)$ obtained in this limit are computed 
explicitly.  The difference equation for $Z_{\frD}(X)$ 
is derived, and confirmed to agree with the result 
of Dunin-Barkowski et al. \cite{DBMNPS13}. 
Section 5 is devoted to Fay-type bilinear equations. 
A three-term bilinear equation plays a central role here.  
The bilinear equation for $Z(\bst)$ is shown to turn 
into a similar bilinear equation for $Z_{\frD}(\bsT)$ 
as $R \to 0$.  The corresponding results for 
$Z(\bst,s)$ and $Z_{\frD}(\bsT,s)$ are presented in Appendix.  
Section 6 concludes this paper.

\section{Melting crystal model}

\subsection{Partition function of 3D Young diagrams}

The partition function of the simplest melting crystal model 
with a single parameter $q$ is the sum 
\beq
  Z = \sum_{\pi\in\calP\calP} q^{|\pi|} 
  \label{Z-PPsum}
\eeq
of the Boltzmann weight $q^{|\pi|}$ over the set $\calP\calP$ 
of all plane partitions.  The plane partition $\pi 
= (\pi_{ij})_{i,j=1}^\infty$ represents a 3D Young diagram 
that consists of stacks of unit cubes of height $\pi_{ij}$ 
put on the unit squares $[i-1,i]\times[j-1,j]$ of the plane. 
$|\pi|$ denotes the volume of the 3D Young diagram:
\[
  |\pi| = \sum_{i,j=1}^\infty \pi_{ij}. 
\]

By the method of diagonal slicing \cite{ORV03}, 
one can convert the sum (\ref{Z-PPsum}) over $\calP\calP$ 
to a sum over the set $\calP$ of all ordinary partitions 
$\lambda = (\lambda_i)_{i=1}^\infty$ as 
\beq
  Z = \sum_{\lambda\in\calP}s_\lambda(q^{-\rho})^2, 
  \label{Z-Psum}
\eeq
where $s_\lambda(q^{-\rho})$ is the special value 
(a kind of principal specialization) of the infinite-variate 
Schur function $s_\lambda(\bsx)$, $\bsx = (x_1,x_2,\ldots)$, at 
\beqnn
   \bsx = q^{-\rho} = (q^{1/2},q^{3/2},\ldots,q^{i-1/2},\ldots). 
\eeqnn
Moreover, by the Cauchy identities of Schur functions \cite{Mac-book}, 
one can rewrite the sum (\ref{Z-Psum}) into an infinite product: 
\beqnn
  Z = \prod_{i,j=1}^\infty(1 - q^{i+j-1})^{-1} 
    = \prod_{n=1}^\infty(1 - q^n)^{-n}. 
\eeqnn
This infinite product is known as the MacMahon function. 

The special value $s_\lambda(q^{-\rho})$ has the hook-length formula 
\cite{Mac-book} 
\beq
  s_\lambda(q^{-\rho}) 
  = \frac{q^{-\kappa(\lambda)/4}}
    {\prod_{(i,j)\in\lambda}(q^{-h(i,j)/2} - q^{h(i,j)/2})}, 
  \label{q-hook-formula}
\eeq
where $\kappa(\lambda)$ is the commonly used notation 
\beqnn
  \kappa(\lambda) 
  = 2\sum_{(i,j)\in\lambda}(j - i) 
  = \sum_{i=1}^\infty
    \left((\lambda_i-i+1/2)^2 - (-i+1/2)^2\right), 
\eeqnn
and $h(i,j)$ denote the the hook length 
\beqnn
  h(i,j) = (\lambda_i-j) + (\tp{\lambda}_j-i) + 1 
\eeqnn
of the cell $(i,j)$ in the Young diagram of shape $\lambda$.  
$\tp{\lambda}_i$'s are the parts of the conjugate partition 
$\tp{\lambda}$ that represents the transposed Young diagram. 
Thus $s_\lambda(q^{-\rho})$ is a $q$-deformation of the number 
\beq
  \frac{\dim\lambda}{|\lambda|!} 
  = \frac{1}{\prod_{(i,j)\in\lambda}h_{(i,j)}}, 
  \label{hook-formula}
\eeq
where $|\lambda| = \sum_{i=1}^\infty\lambda_i$,  
and $\dim\lambda$ is the number of standard tableau 
of shape $\lambda$.  (\ref{hook-formula}) plays 
a central role in Gromov-Witten/Hurwitz theory 
of $\CC\PP^1$ \cite{Okounkov00,OP02,Pandharipande02}.

\subsection{Deformation by external potentials}

We now introduce a parameter $Q$ and an infinite set 
of coupling constants $\bst = (t_1,t_2,\ldots)$, 
and deform (\ref{Z-Psum}) as 
\beq
  Z(\bst) = \sum_{\lambda\in\calP}s_\lambda(q^{-\rho})^2
            Q^{|\lambda|}e^{\phi(\bst,\lambda)}, 
  \label{Z(t)-Psum}
\eeq
where 
\beqnn
  \phi(\bst,\lambda) = \sum_{k=1}^\infty t_k\phi_k(\lambda).
\eeqnn
The external potentials $\phi_k(\lambda)$ are defined as 
\beq
  \phi_k(\lambda) 
  = \sum_{i=1}^\infty\left(q^{k(\lambda_i-i+1)} - q^{k(-i+1)}\right). 
  \label{phi_k}
\eeq
The sum on the right hand side of (\ref{phi_k}) is a finite sum 
because only a finite number of $\lambda_i$'s are non-zero.  
These potentials are $q$-analogues of the so called 
Casimir invariants of the infinite symmetric group $S_\infty$, 
which we shall encounter in the 4D limit.  

The following fact is a consequence of our previous work 
on the melting crystal model \cite{NT07}.  
This fact is by no means obvious from the definition 
of $Z(\bst)$, and explained with the aid of algebraic structures 
hidden behind a fermionic expression of this partition function. 
Moreover, we shall need further refinements of this statement 
to consider an associated quantum spectral curve.  

\begin{theorem}
$Z(\bst)$ is a tau function of the KP hierarchy 
with time variables $\bst = (t_1,t_2,\ldots)$.  
\end{theorem}

Actually, $Z(\bst)$ is a member of a set 
of functions $Z(\bst,s)$, $s \in \ZZ$, considered 
in our previous work: 
\beq
  Z(\bst,s) = \sum_{\lambda\in\calP}s_\lambda(q^{-\rho})^2 
              Q^{|\lambda|+s(s+1)/2}e^{\phi(\bst,s,\lambda)},
  \label{Z(t,s)-Psum}
\eeq
where
\begin{gather*}
  \phi(\bst,s,\lambda) = \sum_{k=1}^\infty t_k\phi_k(s,\lambda),\\
  \phi_k(s,\lambda) 
  = \sum_{k=1}^\infty\left(q^{k(\lambda_i-i+1+s)} - q^{k(-i+1+s)}\right)
     + \frac{1-q^{ks}}{1-q^k}q^k.
\end{gather*}
The $s$-dependent formulation stems from a fermionic interpretation 
of $Z(\bst)$ that we shall review below.  
The external potentials $\phi_k(s,\lambda)$ are obtained 
by rearrangement of terms of the formal expression 
\beqnn
  \phi_k(s,\lambda) 
  = \sum_{i=1}^\infty q^{k(\lambda_i-i+1+s)} - \sum_{i=1}^\infty q^{k(-i+1)}. 
\eeqnn
This explains the origin of the last term of $\phi(s,\lambda)$: 
\begin{align}
  \frac{1-q^{ks}}{1-q^k}q^k 
  &= \sum_{i=1}^\infty q^{k(-i+1+s)} - \sum_{i=1}^\infty q^{k(-i+1)} \notag\\
  &= \begin{cases}
    q^k + \cdots + q^{ks}& \text{for $s > 0$},\\
    0 & \text{for $s = 0$},\\
    - 1 - q^{-k} - \cdots - q^{k(s+1)}& \text{for $s < 0$}. 
    \end{cases}
  \label{phi-corr}
\end{align}
It is shown in our previous work \cite{NT07} that $Z(\bst,s)$ 
is essentially (i.e., up to a simple multiplier and rescaling 
of the time variables) a tau function of the 1D Toda hierarchy.  
This implies, in particular, that $Z(\bst,s)$ is also 
a collection of tau functions of the KP hierarchy 
labelled by $s$ \cite{UT84}.

\subsection{Fermionic expression of partition function}

Let $\psi_n,\psi^*_n$, $n\in\ZZ$, be the creation--annihilation 
operators\footnote{For the sake of convenience, 
as in our previous work \cite{NT07}, we label 
these operators with integers rather than half-integers. 
The free fermion fields are defined as 
$\psi(z) = \sum_{n\in\ZZ}\psi_nz^{-n-1}$ 
and $\psi^*(z) = \sum_{n\in\ZZ}\psi^*_nz^{-n}$.} 
of 2D charged free fermion theory with the anti-commutation relations 
\beqnn
  \psi_m\psi^*_n + \psi^*_n\psi_m = \delta_{m+n,0},\quad 
  \psi_m\psi_n + \psi_n\psi_m = 0,\quad 
  \psi^*_m\psi^*_n + \psi^*_n\psi^*_m = 0, 
\eeqnn
and $|0\rangle$, $\langle 0|$, $s \in \ZZ$, 
the vacuum vectors of the fermionic Fock and dual Fock spaces 
that satisfy the vacuum conditions 
\begin{gather*}
  \psi^*_n|0\rangle = 0 \quad \text{for $n > 0$},\quad 
  \psi_n|0\rangle = 0 \quad \text{for $n \ge 0$},\\
  \langle 0|\psi_n = 0 \quad \text{for $n < 0$},\quad 
  \langle 0|\psi^*_n = 0 \quad \text{for $n \le 0$}. 
\end{gather*}
The charge-$0$ sectors of the Fock spaces are spanned 
by the excited states $|\lambda\rangle$, $\langle\lambda|$, 
$\lambda\in\calP$: 
\begin{align*}
  |\lambda\rangle &= \psi_{-\lambda_1}\psi_{-\lambda_2+1}\cdots
    \psi_{-\lambda_n+n-1}\psi^*_{-n+1}\cdots\psi^*_{-1}\psi^*_0|0\rangle,\\
  \langle \lambda| &= \langle 0|\psi_0\psi_1\cdots\psi_{n-1}
    \psi^*_{\lambda_n-n+1}\cdots\psi^*_{\lambda_2-1}\psi^*_{\lambda_1},
\end{align*}
where $n$ is chosen so that $\lambda_i = 0$ for $i > n$.  
In particular, $|\emptyset\rangle$ and $\langle\emptyset|$ 
agree with the vacuum states.    
The charge-$s$ sectors of the Fock spaces are spanned 
by similar vectors $|s,\lambda\rangle$, $\langle s,\lambda|$, 
$\lambda \in \calP$. 

The fermionic expression of the aforementioned partition functions 
employs the normally ordered fermion bilinears 
\begin{gather*}
  L_0 = \sum_{n\in\ZZ}n{:}\psi_{-n}\psi^*_n{:},\quad 
  K = \sum_{n\in\ZZ}(n-1/2)^2{:}\psi_{-n}\psi^*_n{:},\\
  H_k = \sum_{n\in\ZZ}q^{kn}{:}\psi_{-n}\psi^*_n{:},\quad 
  J_k = \sum_{n\in\ZZ}{:}\psi_{k-n}\psi^*_n{:},\\
  {:}\psi_m\psi^*_n{:} 
  = \psi_m\psi^*_n - \langle 0|\psi_m\psi^*_n|0\rangle, 
\end{gather*}
the vertex operators \cite{ORV03,YB08}
\beqnn
  \Gamma_{\pm k}(x) 
  = \exp\left(\sum_{k=1}^\infty\frac{x^k}{k}J_{\pm k}\right),\quad 
  \Gamma'_{\pm k}(x) 
  = \exp\left(- \sum_{k=1}^\infty\frac{(-x)^k}{k}J_{\pm k}\right),\quad 
\eeqnn
and their multi-variate extensions 
\beqnn
  \Gamma_{\pm k}(x_1,x_2,\ldots) = \prod_{i\ge 1}\Gamma_{\pm k}(x_i),\quad
  \Gamma'_{\pm k}(x_1,x_2,\ldots) = \prod_{i\ge 1}\Gamma'_{\pm k}(x_i). 
\eeqnn
The action of these operators on the fermionic Fock space 
leaves the charge-$0$ sector invariant.  
$L_0$, $K$ and $H_k$ are diagonal with respect to $|\lambda\rangle$'s: 
\beq
  \langle\lambda|L_0|\mu\rangle = |\lambda|\delta_{\lambda\mu},\quad 
  \langle\lambda|K|\mu\rangle = \kappa(\lambda)\delta_{\lambda\mu},\quad
  \langle\lambda|H_k|\mu\rangle = \phi_k(\lambda)\delta_{\lambda\mu}. 
  \label{<..L_0..>etc}
\eeq
The matrix elements of the vertex operators are 
the skew Schur functions $s_{\lambda/\mu}(\bsx)$, 
$\bsx = (x_1,x_2,\ldots)$: 
\begin{gather}
  \langle\lambda|\Gamma_{-}(\bsx)|\mu\rangle 
  = \langle\mu|\Gamma_{+}(\bsx)|\lambda\rangle 
  = s_{\lambda/\mu}(\bsx),
  \label{<..Gamma..>}\\
  \langle\lambda|\Gamma'_{-}(\bsx)|\mu\rangle 
  = \langle\mu|\Gamma'_{+}(\bsx)|\lambda\rangle 
  = s_{\tp{\lambda}/\tp{\mu}}(\bsx). 
  \label{<..Gamma'..>}
\end{gather}

One can use these building blocks to rewrite 
the combinatorial definition (\ref{Z(t)-Psum}) of $Z(\bst)$ as 
\beq
  Z(\bst) = \langle 0|\Gamma_{+}(q^{-\rho})Q^{L_0}
            e^{H(\bst)}\Gamma_{-}(q^{-\rho})|0\rangle, 
  \label{Z(t)=<..e^H..>}
\eeq
where 
\beqnn
  H(\bst) = \sum_{k=1}^\infty t_kH_k. 
\eeqnn
As shown in our previous work with the aid of 
symmetries of a quantum torus algebra \cite{NT07} , 
(\ref{Z(t)=<..e^H..>}) can be converted to the following form.  
This implies that $Z(\bst)$ is a tau function of the KP hierarchy. 

\begin{theorem}
\beq
  Z(\bst) = \exp\left(\sum_{k=1}^\infty\frac{q^kt_k}{1-q^k}\right) 
    \langle 0|\exp\left(\sum_{k=1}^\infty(-1)^kq^{k/2}t_kJ_k\right)
    g_1|0\rangle, 
  \label{Z(t)=<..g_1..>}
\eeq
where 
\beq
  g_1 = q^{K/2}\Gamma_{-}(q^{-\rho})\Gamma_{+}(q^{-\rho})Q^{L_0}
        \Gamma_{-}(q^{-\rho})\Gamma_{+}(q^{-\rho})q^{K/2}. 
  \label{g_1}
\eeq
\end{theorem}

One can rewrite (\ref{Z(t)=<..g_1..>}) further to clarify 
its characteristic as a tau function of the KP hierarchy.  
Firstly, the exponential prefactor on the right hand side 
can be taken inside the vev as 
\beq
  Z(\bst) = \langle 0|\exp\left(\sum_{k=1}^\infty(-1)^kq^{k/2}t_kJ_k\right)
    g_2|0\rangle 
  \label{Z(t)=<..g_2..>}
\eeq
with 
\beqnn
  g_2 = \exp\left(\sum_{k=1}^\infty\frac{(-1)^kq^{k/2}}{k(1-q^k)}J_{-k}\right)g_1.
\eeqnn
This is a consequence of the commutation relation 
\beqnn
  [J_m,J_n] = m\delta_{m+n,0}
\eeqnn
among $J_n$'s that span the $U(1)$ current (or Heisenberg) algebra. 
Remarkably, the operator generated in front of $g_1$, too, 
is related to a vertex operator as 
\beqnn
  \exp\left(\sum_{k=1}^\infty\frac{(-1)^kq^{k/2}}{k(1-q^k)}J_{-k}\right)
  = \Gamma'_{-}(q^{-\rho})^{-1}. 
\eeqnn
Thus $g_2$ can be expressed as 
\beq
  g_2 = \Gamma'_{-}(q^{-\rho})^{-1}
        q^{K/2}\Gamma_{-}(q^{-\rho})\Gamma_{+}(q^{-\rho})Q^{L_0}
        \Gamma_{-}(q^{-\rho})\Gamma_{+}(q^{-\rho})q^{K/2}. 
  \label{g_2}
\eeq
Secondly, the multipliers $(-1)^kq^{k/2}$ of $t_k$'s can be removed 
by the scaling relation 
\beqnn
  \sum_{k=1}^\infty (-1)^kq^{k/2}t_kJ_k
  = (-q^{1/2})^{-L_0}\cdot\sum_{k=1}^\infty t_kJ_k\cdot(-q^{1/2})^{L_0}.
\eeqnn
(\ref{Z(t)=<..g_2..>}) thereby turns into the more standard expression 
\beq
  Z(\bst) = \langle 0|\exp\left(\sum_{k=1}^\infty t_kJ_k\right)g_3|0\rangle, 
  \quad g_3= (-q^{1/2})^{L_0}g_2, 
  \label{Z(t)=<..g_3..>}
\eeq
as a tau function of the KP hierarchy \cite{JM83,MJD-book}. 

\begin{remark}
As we shall see in the next section, 
one can simplify the operator $g_2$ to 
\beq
  g = \Gamma'_{-}(q^{-\rho})^{-1}q^{K/2}\Gamma_{-}(q^{-\rho})\Gamma_{-}(Qq^{-\rho}) 
  \label{g}
\eeq
without changing the associated tau function of the KP hierarchy. 
$g_1$ is defined in the somewhat complicated form (\ref{g_1}) 
to enjoy the algebraic relations 
\beq
  J_kg_1 = g_1J_k \quad \text{for $k = 1,2,\ldots$}. 
  \label{g_1-symmetry}
\eeq
These relations ensure that the associated tau function 
of the 2D Toda hierarchy reduces to a tau function 
of the 1D Toda hierarchy \cite{NT07}. 
\end{remark}

\begin{remark}
In the previous work \cite{NT07}, we used the operator 
\beqnn
  W_0 = \sum_{n\in\ZZ}n^2{:}\psi_{-n}\psi^*_n{:}
\eeqnn
in place of $K$.  Accordingly, the fermionic expression 
of the partition functions presented therein takes 
a slightly different form.  This does not affect 
the essential part of the fermionic expression. 
\end{remark}

\section{Quantum spectral curve}

\subsection{Single-variate specialization}

Let $Z(x)$ denote the single-variate specialization of $Z(\bst)$ 
obtained by substituting 
\beq
  t_k = - \frac{q^{-k/2}x^k}{k},\quad k = 1,2,\ldots. 
  \label{t-x-rel}
\eeq
The combinatorial definition (\ref{Z(t)-Psum}) of $Z(\bst)$ 
and its fermionic expressions (\ref{Z(t)=<..g_1..>}) and 
(\ref{Z(t)=<..g_2..>}) are accordingly specialized 
as follows. 

\begin{lemma}
\beq
  Z(x) = \sum_{\lambda\in\calP}s_\lambda(q^{-\rho})^2Q^{|\lambda|}
     \prod_{i=1}^\infty\frac{1 - q^{\lambda_i-i+1/2}x}{1 - q^{-i+1/2}x}.
  \label{Z(x)-Psum}
\eeq
\end{lemma}

\begin{proof}
Substituting (\ref{t-x-rel}) for $\phi(\bst,\lambda)$ yields 
\begin{align*}
  \phi(\bst,\lambda) 
  &= - \sum_{i,k=1}^\infty\frac{q^{-k/2}x^k}{k}
       \left(q^{k(\lambda_i-i+1)} - q^{k(-i+1)}\right)\\
  &= \sum_{i=1}^\infty\left(\log(1 - q^{\lambda_i-i+1/2}x) 
         - \log(1 - q^{-i+1/2}x)\right), 
\end{align*}
hence 
\beqnn
  e^{\phi(\bst,\lambda)} 
  = \prod_{i=1}^\infty\frac{1 - q^{\lambda_i-i+1/2}x}{1 - q^{-i+1/2}x}. 
\eeqnn
\end{proof}

\begin{lemma}
\beq
  Z(x) = \prod_{i=1}^\infty(1 - q^{i-1/2}x)
         \cdot\langle 0|\Gamma'_{+}(x)g_1|0\rangle 
       = \langle 0|\Gamma'_{+}(x)g_2|0\rangle. 
  \label{Z(x)=<..g_2..>}
\eeq
\end{lemma}

\begin{proof}
Substituting (\ref{t-x-rel}) in (\ref{Z(t)=<..g_1..>}) 
and (\ref{Z(t)=<..g_2..>}) yields
\beqnn
  \exp\left(\sum_{k=1}^\infty\frac{q^kt_k}{1-q^k}\right) 
  = \exp\left(- \sum_{i,k=1}^\infty\frac{q^{k(i-1/2)}x^k}{k}\right) 
  = \prod_{i=1}^\infty(1 - q^{i-1/2}x) 
\eeqnn
(cf. the computation in the proof of the previous lemma) and 
\beqnn
  \exp\left(\sum_{k=1}^\infty(-1)^kq^{k/2}t_kJ_k\right)
  = \exp\left(- \sum_{k=1}^\infty\frac{(- x)^k}{k}J_k\right)
  = \Gamma'_{+}(x). 
\eeqnn
\end{proof}

As we shall see in the next section, 
the combinatorial expression (\ref{Z(x)-Psum}) of $Z(x)$ 
has a desirable form from which one can derive the equation 
of quantum curve of Dunin-Barkowski et al. \cite{DBMNPS13}. 
To apply the method of our previous work \cite{TN16}, 
however, it is more convenient to have $\Gamma_{+}(x)$ 
rather than $\Gamma'_{+}(x)$ in the fermionic expression 
(\ref{Z(x)=<..g_2..>}) of $Z(x)$.  
This problem can be settled by the following 
transformation rule of matrix elements of fermionic operators 
under conjugation of partitions \cite{YB08}: 

\begin{lemma}
\label{conj-lemma}
\begin{gather*}
  \langle\lambda|L_0|\lambda\rangle 
  = \langle\tp{\lambda}|L_0|\tp{\lambda}\rangle,\quad 
  \langle\lambda|K|\lambda\rangle 
  = - \langle\tp{\lambda}|K|\tp{\lambda}\rangle,\\
  \langle\lambda|\Gamma_{\pm}(\bsx)|\mu\rangle 
  = \langle\tp{\lambda}|\Gamma'_{\pm}(\bsx)|\tp{\mu}\rangle. 
\end{gather*}
\end{lemma}

\begin{proof}
These identities are consequences of (\ref{<..L_0..>etc}), 
(\ref{<..Gamma..>}), (\ref{<..Gamma'..>}) and 
the following property of $\kappa(\lambda)$: 
\beqnn
  \kappa(\tp{\lambda}) = - \kappa(\lambda).
\eeqnn
\end{proof}

We can apply this rule to $\Gamma'_{+}(x)$ and 
the building blocks of $g_2$ to rewrite (\ref{Z(x)=<..g_2..>}) as 
\beq
  Z(x) = \langle 0|\Gamma_{+}(x)g_2'|0\rangle, 
  \label{Z(x)=<..g_2'..>}
\eeq
where 
\beq
  g_2' = \Gamma_{-}(q^{-\rho})^{-1}
        q^{-K/2}\Gamma'_{-}(q^{-\rho})\Gamma'_{+}(q^{-\rho})Q^{L_0}
        \Gamma'_{-}(q^{-\rho})\Gamma'_{+}(q^{-\rho})q^{-K/2}. 
  \label{g_2'}
\eeq

Having obtained the fermionic expression (\ref{Z(x)=<..g_2'..>}) 
containing $\Gamma_{+}(x)$, we now remove the other $\Gamma_{+}$'s 
from (\ref{Z(x)=<..g_2'..>}).  This is the last step 
for applying the method of out previous work \cite{TN16}. 

\begin{lemma}
\beq
  Z(x) = \prod_{n=1}^\infty(1 - Qq^n)^{-n}
         \cdot\langle 0|\Gamma_{+}(x)g'|0\rangle, 
  \label{Z(x)=<..g..>}
\eeq
where 
\beq
  g' = \Gamma_{-}(q^{-\rho})^{-1}q^{-K/2}
      \Gamma'_{-}(q^{-\rho})\Gamma'_{-}(Qq^{-\rho}). 
  \label{g'}
\eeq
\end{lemma}

\begin{proof} 
Since the rightmost two factors $\Gamma'_{+}(q^{-\rho})q^{-K/2}$ 
of (\ref{g_2'}) act on the vacuum vector trivially as 
\beqnn
  \Gamma'_{+}(q^{-\rho})q^{-K/2}|0\rangle = |0\rangle, 
\eeqnn
one can remove these operators from (\ref{Z(x)=<..g_2'..>}). 
Moreover, one can use the scaling relation 
\beqnn
  \Gamma'_{\pm}(x_1,x_2,\ldots)Q^{L_0} 
  = Q^k\Gamma'_{\pm}(Q^{\pm 1}x_1,Q^{\pm 1}x_2,\ldots) 
\eeqnn
and the commutation relation
\beqnn
  \Gamma'_{+}(x_1,x_2,\ldots)\Gamma'_{-}(y_1,y_2,\ldots)
  = \prod_{i,j=1}^\infty(1 - x_iy_j)^{-1}\cdot 
    \Gamma'_{-}(y_1,y_2,\ldots)\Gamma'_{+}(x_1,x_2,\ldots)
\eeqnn
of the vertex operators \cite{ORV03,YB08} 
to rewrite the product of the three operators 
in the middle of (\ref{g_2'}) as 
\begin{align*}
  \Gamma'_{+}(q^{-\rho})Q^{L_0}\Gamma'_{-}(q^{-\rho}) 
  &= Q^{L_0}\Gamma'_{+}(Qq^{-\rho})\Gamma'_{-}(q^{-\rho}) \\
  &= \prod_{n=1}^\infty(1 - Qq^n)^{-n}\cdot Q^{L_0} 
    \Gamma'_{-}(q^{-\rho})\Gamma'_{+}(Qq^{-\rho})\\
  &= \prod_{n=1}^\infty(1 - Qq^n)^{-n}\cdot 
     \Gamma'_{-}(Qq^{-\rho})Q^{L_0}\Gamma'_{+}(Qq^{-\rho}). 
\end{align*}
The two factors $Q^{L_0}\Gamma'_{+}(Qq^{-\rho})$ 
in the last line hit the vacuum vector $|0\rangle$ 
and disappear.  What remains are 
the constant $\prod_{n=1}^\infty(1 - Qq^n)^{-n}$ 
and the operator $g'$.  
\end{proof}

\begin{remark}
Note that the operator $g'$ amounts to the operator $g$ 
defined in (\ref{g'}).  Also note that the foregoing computations 
are actually a proof of the identity 
\beq
  g_2'|0\rangle = \prod_{n=1}^\infty(1 - Qq^n)^{-n}\cdot g'|0\rangle
  \label{g_2'|0>=const.g'|0>}
\eeq
of vectors in the fermionic Fock space. 
\end{remark}

\subsection{Generating operator of admissible basis}

We now borrow the idea of generating operators from the work 
of Alexandrov et al. \cite{Alexandrov1404,ALS1512}. 
A point of the Sato Grassmannian can be represented 
by a linear subspace $W$ of the space $V = \CC((x))$ 
of formal Laurent series \cite{SS83,SW85}.  
The generating operator is a linear automorphism $G$ of $V$ 
that maps $W_0 = \mathrm{Span}\{x^{-j}\}_{j=0}^\infty$ 
to $W$, so that an admissible basis $\{\Phi_j(x)\}_{j=0}^\infty$ 
of $W$ can be expressed as 
\beq
  \Phi_j(x) = Gx^{-j}. 
  \label{basis}
\eeq

In the fermionic formalism of the KP hierarchy \cite{JM83,MJD-book}, 
$W$ corresponds to a vector $|W\rangle$ 
of the fermionic Fock space. 
The associated tau function can be defined as 
\beqnn
  \tau(\bst) 
  = \langle 0|\exp\left(\sum_{k=1}^\infty t_kJ_k\right)|W\rangle. 
\eeqnn
Its special value at 
\beq
  \bst = [x] = (x, x^2/2,\ldots,x^k/k,\ldots) 
  \label{t=[x]}
\eeq
is related to the first member $\Phi_0(x)$ 
of an admissible basis of $W$ as 
\beq
  \tau([x]) = \langle 0|\Gamma_{+}(x)|W\rangle = C\Phi_0(x), 
  \label{tau(x)} 
\eeq
where $C$ is a nonzero constant.  

If $|W\rangle$ is generated from the vacuum vector $|0\rangle$ 
by an operator $g$ as 
\beqnn
  |W\rangle = g|0\rangle, 
\eeqnn
and $g$ is a special operator, such as a product 
of vertex operators and particular diagonal operators, 
then one can find $G$ rather easily from $g$ 
by the correspondence 
\beq
  L_0 \longleftrightarrow D = x\dfrac{d}{dx},\quad 
  K \longleftrightarrow \left(D - \frac{1}{2}\right)^2,\quad 
  J_k \longleftrightarrow x^{-k},\quad 
  \text{etc.},  
  \label{fb-do}
\eeq
between fermion bilinears and differential operators.  
This is the way how Alexandrov et al. derived 
the generating operator for various types 
of Hurwitz numbers \cite{ALS1512}. 
We did similar computations for tau functions 
in topological string theory \cite{TN16}. 
We apply the same method to the operator $g'$ of (\ref{g'}). 

It is now straightforward to find the generating operator $G$ 
of the subspace $W \subset V$ determined by the operator $g'$ 
of (\ref{g'}).  According to (\ref{fb-do}), 
$q^{-K/2}$ corresponds to a differential operator of infinite order: 
\beqnn
  q^{-K/2} \longleftrightarrow q^{-(D-1/2)^2/2}. 
\eeqnn
The three vertex operators in $g'$ amount to multiplication operators: 
\begin{gather*}
  \Gamma_{-}(q^{-\rho})^{-1} 
  \longleftrightarrow 
  \exp\left(- \sum_{i,k=1}^\infty\frac{q^{k(i-1/2)}x^k}{k}\right) 
  = \prod_{i=1}^\infty(1-q^{i-1/2}x),\\
  \Gamma'_{-}(q^{-\rho}) 
    \longleftrightarrow \prod_{i=1}^\infty(1+q^{i-1/2}x),\\
  \Gamma'_{-}(Qq^{-\rho}) 
    \longleftrightarrow \prod_{i=1}^\infty(1+Qq^{i-1/2}x). 
\end{gather*}
The generating operator is given by a product  
of these operators as follows. 

\begin{theorem}
The generating operator $G$ for the subspace 
$W \subset V$ determined by the operator $g'$ 
of (\ref{g'}) can be expressed as 
\beq
  G = \prod_{i=1}^\infty(1-q^{i-1/2}x)\cdot q^{-(D-1/2)^2/2} 
      \cdot\prod_{i=1}^\infty(1+q^{i-1/2}x)(1+Qq^{i-1/2}x). 
  \label{G}
\eeq
\end{theorem}

\subsection{Derivation of quantum spectral curve}

Since the structure of the generating operator (\ref{G}) resembles 
those of tau functions in topological string theory \cite{TN16}, 
we define the Kac--Schwarz operator $A$ in essentially the same form, 
\footnote{This operator amounts to the inverse $A^{-1}$ 
of the Kac--Schwarz operator $A$ considered therein.}
namely, 
\beqnn
  A = G\cdot q^D\cdot G^{-1}. 
\eeqnn
The members $\Phi_j(x)$ of the admissible basis (\ref{basis}) 
thereby satisfy the linear equations 
\beqnn
  A\Phi_j(x) = q^{-j}\Phi_j(x). 
\eeqnn
In particular, the equation 
\beqnn
  (A - 1)\Phi_0(x) = 0 
\eeqnn
for $\Phi_0(x)$ (equivalently, $\tau([x])$) represents  
the quantum spectral curve.  As we show below, 
$A$ is a $q$-difference operator of finite order. 

\begin{lemma}
\begin{align}
  A &= \left(1 + q^{1/2}xq^{-D}(1-q^{1/2}x)^{-1}\right)
      \left(1 + Qq^{1/2}xq^{-D}(1-q^{1/2}x)^{-1}\right) \notag\\
    &\quad\mbox{}\times (1 - q^{1/2}x)q^D. 
  \label{A-prod}
\end{align}
\end{lemma}

\begin{proof}
One can compute $A = G\cdot q^D\cdot G^{-1}$ step by step. 
The first step is to apply the last infinite product 
of (\ref{G}) and its inverse to $q^D$.  
This can be carried out with the aid of the operator identity 
\beqnn
  q^D\cdot x = qxq^D 
\eeqnn
as follows: 
\begin{align*}
  &\prod_{i=1}^\infty(1+q^{i-1/2}x)(1+Qq^{i-1/2}x)\cdot q^D 
   \cdot\prod_{i=1}^\infty(1+Qq^{i-1/2}x)^{-1}(1+q^{i-1/2}x)^{-1}\\
  &= \prod_{i=1}^\infty(1+q^{i-1/2}x)(1+Qq^{i-1/2}x)\cdot 
     \prod_{i=1}^\infty(1+Qq^{i+1/2}x)^{-1}(1+q^{i+1/2}x)^{-1}\cdot q^D\\
  &= (1 + q^{1/2}x)(1 + Qq^{1/2}x)q^D. 
\end{align*}
The next step is to apply $q^{-(D-1/2)^2/2}$ and its inverse 
to the last operator.  This can be achieved by the identity 
\beqnn
  q^{-(D-1/2)^2/2}\cdot x\cdot q^{(D-1/2)^2/2} = xq^{-D} 
\eeqnn
as follows: 
\begin{align*}
  &q^{-(D-1/2)^2/2}\cdot (1 + q^{1/2}x)(1 + Qq^{1/2}x)q^D\cdot q^{(D-1/2)^2/2}\\
  &= (1 + q^{1/2}xq^{-D})(1 + Qq^{1/2}xq^{-D})q^D. 
\end{align*}
The first infinite product of (\ref{G}) and its inverse 
transform $q^{-D}$ and $q^D$ in this operator product as 
\begin{gather*}
  \prod_{i=1}^\infty(1-q^{i-1/2}x)\cdot q^{-D}
  \cdot\prod_{i=1}^\infty(1-q^{i-1/2}x)^{-1} 
  = q^{-D}(1 - q^{1/2}x)^{-1}, \\
  \prod_{i=1}^\infty(1-q^{i-1/2}x)\cdot q^D
  \cdot\prod_{i=1}^\infty(1-q^{i-1/2}x)^{-1} 
  = (1 - q^{1/2}x)q^D. 
\end{gather*}
Thus one obtains the result shown in (\ref{A-prod}). 
\end{proof}

Let us expand (\ref{A-prod}) and move $q^{\pm D}$ 
in each term to the right end.  The outcome reads 
\beq
  A = (1-q^{1/2}x)q^D + q^{1/2}x + Qq^{1/2}x 
      + Qx^2(1-q^{-1/2}x)^{-1}q^{-D}. 
  \label{A-sum}
\eeq
We are thus led to the following final expression 
of the quantum spectral curve of the melting crystal model. 

\begin{theorem}
$Z(x)$ satisfies the equation 
\beq
  (A - 1)Z(x) = 0 
  \label{(A-1)Z(x)=0}
\eeq
with respect to the $q$-difference operator (\ref{A-sum}). 
\end{theorem}

\section{Prescription for 4D limit}

The 4D limit of the partition function $Z(\bszero)$ 
at $\bst = \bszero$ is achieved by setting the parameters as 
\beq
  q = e^{-R\hbar},\quad Q = (R\Lambda)^2 
  \label{qQ-4Dlimit}
\eeq
and letting $R \to 0$ \cite{NT07}. $R$ is the radius 
of the fifth dimension of $\RR^4\times S^1$ in which 
SUSY Yang-Mills theory lives \cite{Nekrasov96}, 
$\hbar$ is a parameter of the self-dual $\Omega$ background,  
and $\Lambda$ is an energy scale of 4D $\calN = 2$ 
SUSY Yang-Mills theory \cite{Nekrasov02,NO03}. 
The definition of 4D limit of $Z(x)$ and $Z(\bst)$ needs 
$R$-dependent transformations of $x$ and $\bst$. 

\subsection{4D limit of $Z(x)$ and quantum spectral curve}

Alongside the substitution (\ref{qQ-4Dlimit}) of parameters, 
we transform the variable $x$ to a new variable $X$ as 
\beq
  x = x(X,R) = e^{R(X-\hbar/2)}. 
  \label{x-X-rel}
\eeq
As it turns out below, both the combinatorial expression 
(\ref{Z(x)-Psum}) and the $q$-difference equation 
(\ref{(A-1)Z(x)=0}) of $Z(x)$ behave nicely as $R \to 0$ 
under this $R$-dependent transformation of $x$.  

\begin{lemma}
\label{Z4D(X)-lemma}
\beqnn
  \lim_{R\to 0}Z(x(X,R)) = Z_{\frD}(X),
\eeqnn
where 
\beq
  Z_{\frD}(X) = \sum_{\lambda\in\calP}
     \left(\frac{\dim\lambda}{|\lambda|}\right)^2 
     \left(\frac{\Lambda}{\hbar}\right)^{2|\lambda|} 
     \prod_{i=1}^\infty\frac{X-(\lambda_i-i+1)\hbar}{X-(-i+1)\hbar}.
  \label{Z4D(X)-Psum}
\eeq
\end{lemma}

\begin{proof}
As $R \to 0$ under the $R$-dependent transformations 
(\ref{qQ-4Dlimit}) and (\ref{x-X-rel}), 
the building blocks of (\ref{Z(x)-Psum}) behave as 
\begin{align*}
  s_\lambda(q^{-\rho})^2 
  &= \left(\frac{\dim\lambda}{|\lambda|}\right)^2(R\hbar)^{-2|\lambda|}
     (1 + O(R)),\\
  Q^{|\lambda|} &= (R\Lambda)^{2|\lambda|},\\
  \frac{1 - q^{\lambda_i-i+1/2}x(X,R)}{1 - q^{-i+1/2}x(X,R)} 
  &= \frac{X - (\lambda_i-i+1)\hbar}{X - (-i+1)\hbar}(1 + O(R)).
\end{align*}
Note that the hook-length formulae (\ref{q-hook-formula}) 
and (\ref{hook-formula}) are used in the derivation 
of the first line above.  
\end{proof}

\begin{lemma}
\beqnn
  A-1 = \left(- (X-\hbar)(e^{-\hbar d/dX}-1) 
           - \frac{\Lambda^2}{X}e^{\hbar d/dX}\right)R + O(R^2). 
\eeqnn
\end{lemma}

\begin{proof}
(\ref{A-sum}) implies that $A-1$ can be expressed as 
\beqnn
  A-1 = (1 - q^{1/2}x)(q^D - 1) + Qq^{1/2}x 
        + Qx^2(1 - q^{-1/2}x)^{-1}q^{-D}. 
\eeqnn
As $R \to 0$ under the transformations (\ref{qQ-4Dlimit}) 
and (\ref{x-X-rel}), each term of this expression 
behaves as follows: 
\begin{align*}
  1 - q^{1/2}x &= - R(X-\hbar) + O(R^2),\\
  q^{\pm D} &= e^{\mp\hbar d/dX},\\
  Qq^{1/2}x &= O(R^2),\\
  Qx^2(1 - q^{-1/2}x)^{-1}q^{-D} 
  &= - \frac{\Lambda^2}{X}R + O(R^2). 
\end{align*}
\end{proof}

As a consequence of the foregoing two facts, we obtain 
the following difference equation for $Z_{\frD}(X)$. 

\begin{theorem}
$Z_{\frD}(X)$ satisfies the difference equation 
\beq
  \left((X-\hbar)(e^{-\hbar d/dX} - 1)
    + \frac{\Lambda^2}{X}e^{\hbar d/dX}\right)Z_{\frD}(X) = 0. 
  \label{Z4D(X)-diffeq}
\eeq
\end{theorem}

By the shift$X \to X + \hbar$ of $X$, 
(\ref{Z4D(X)-diffeq}) turns into the equation 
\beqnn
  \left(X(e^{-\hbar d/dX} - 1) 
    + \frac{\Lambda^2}{X+\hbar}e^{\hbar d/dX}
  \right)Z_{\frD}(X+\hbar) = 0,
\eeqnn
which agrees with the equation derived 
by Dunin-Barkowski et al. \cite{DBMNPS13}. 
Moreover, as they found, this equation can be converted 
to the simpler form 
\beq
  \left(e^{-\hbar d/dX} + \Lambda^2e^{\hbar d/dX} - X\right)\Psi(X) = 0 
  \label{CP1qsc-eq}
\eeq
by the gauge transformation 
\beqnn
  \Psi(X) 
  = \exp\left(B\left(-\hbar\frac{d}{dX}\right)\frac{X-X\log X}{\hbar}\right) 
    Z_{\frD}(X+\hbar),
\eeqnn
where $B(t)$ is the generating function 
\beqnn
  B(t) = \frac{t}{e^t - 1} 
\eeqnn
of the Bernoulli numbers.  It is this equation (\ref{CP1qsc-eq}) 
that is identified by Dunin-Barkowski et al. \cite{DBMNPS13} 
as the equation of quantum spectral curve for Gromov-Witten theory 
of $\CC\PP^1$.  Its classical limit 
\beqnn
  y^{-1} + y - x = 0 
\eeqnn
as $\hbar \to 0$ (with $\Lambda$ normalized to $1$) 
is the spectral curve of topological recursion 
in this case \cite{DBOSS12,DBSS12,NS11}. 
We have thus re-derived the quantum spectral curve 
of $\CC\PP^1$ from the 4D limit of the melting crystal model.

\subsection{4D limit of $Z(\bst)$} 

As shown in the proof of Lemma \ref{Z4D(X)-lemma}, 
the deformed Boltzmann weight $s_\lambda(q^{-\rho})^2Q^{|\lambda|}$ 
behaves nicely in the limit as $R \to 0$.  
To achieve the 4D limit of $Z(\bst)$, 
we have only to find an appropriate $R$-dependent 
transformation $\bst = \bst(\bsT,R)$ 
to the coupling constants $\bsT = (T_1,T_2,\ldots)$ 
of 4D external potentials $\phi^{\frD}_k(\lambda)$ 
for which the identity 
\beq
  \lim_{R\to 0}\phi(\bst(\bsT,R),\lambda) 
  = \phi_{\frD}(\bsT,\lambda) 
  = \sum_{k=1}^\infty T_k\phi^{\frD}_k(\lambda) 
  \label{phi(t)->phi4D(T)}
\eeq
holds.  In accordance with the commonly adopted setting 
in the literature \cite{MN06}, 
we wish to tune the transformation $\bst = \bst(\bsT,R)$ 
so that $\phi^{\frD}_k(\lambda)$'s take the {\it polynomial\/} form 
\beq
  \phi^{\frD}_k(\lambda) 
  = \sum_{i=1}^\infty\left((\lambda_i-i+1)^k - (-i+1)^k\right). 
  \label{phi4D_k}
\eeq
Thus the problem is how to derive 
these polynomial potentials from the {\it exponential\/} 
potentials (\ref{phi_k}).   
The following is a clue to this problem. 

\begin{lemma}
As $R \to 0$ under the transformation (\ref{qQ-4Dlimit}) 
of the parameters,  
\beq
  \sum_{j=1}^k(-1)^{k-j}\binom{k}{j}\phi_j(\lambda) 
  = \phi^{\frD}_k(\lambda)(-R\hbar)^k + O(R^{k+1}). 
  \label{phi-phi4D-rel}
\eeq
\end{lemma}

\begin{proof}
The difference of the two identities 
\begin{gather*}
  \sum_{j=1}^k(-1)^{k-j}\binom{k}{j}q^{ju} = (q^u-1)^k - (-1)^k,\\
  \sum_{j=1}^k(-1)^{k-j}\binom{k}{j}q^{jv} = (q^v-1)^k - (-1)^k
\end{gather*}
yields the identity 
\begin{align}
  \sum_{j=1}^k(-1)^{k-j}\binom{k}{j}(q^{ju} - q^{jv}) 
  &= (q^u-1)^k - (q^v-1)^k \notag\\
  &= (u^k - v^k)(-R\hbar)^k + O(R^{k+1}). 
     \label{q^(ku)-u^k-rel}
\end{align}
One can derive (\ref{phi-phi4D-rel}) by specializing 
this identity to $u = \lambda_i-i+1$ and $v = -i+1$ 
and summing the outcome over $i = 1,2,\ldots$.  
\end{proof}

(\ref{phi-phi4D-rel}) implies the identity 
\beqnn
  \lim_{R\to 0}\sum_{k=1}^\infty\frac{T_k}{(-R\hbar)^k} 
    \sum_{j=1}^k(-1)^{k-j}\binom{k}{j}\phi_j(\lambda) 
  = \sum_{k=1}^\infty T_k\phi^{\frD}_k(\lambda) 
\eeqnn
for $\phi_k(\lambda)$'s and the potentials shown 
in (\ref{phi4D_k}).  Since 
\beqnn
  \sum_{k=1}^\infty\frac{T_k}{(-R\hbar)^k} 
    \sum_{j=1}^k(-1)^{k-j}\binom{k}{j}\phi_j(\lambda) 
  = \sum_{j=1}^\infty\sum_{k=j}^\infty
      \binom{k}{j}\frac{(-1)^{k-j}T_k}{(-R\hbar)^k}\phi_j(\lambda), 
\eeqnn
one can conclude that the identity (\ref{phi(t)->phi4D(T)}) 
holds if $t_k$'s and $T_k$'s are related by the linear relations 
\beq
  t_j = \sum_{k=j}^\infty\binom{k}{j}\frac{(-1)^{k-j}T_k}{(-R\hbar)^k}. 
  \label{t-T-rel}
\eeq
This gives an $R$-dependent transformation $\bst = \bst(\bsT,R)$ 
that we have sought for.  Note that this is a triangular 
(hence invertible) linear transformation between $\bst$ and $\bsT$. 

Let $Z_{\frD}(\bsT)$ denote the deformed partition function 
\beq
  Z_{\frD}(\bsT)
  = \sum_{\lambda\in\calP}\left(\frac{\dim\lambda}{|\lambda|!}\right)^2 
      \left(\frac{\Lambda}{\hbar}\right)^{2|\lambda|}
      e^{\phi_{\frD}(\bsT,\lambda)} 
  \label{Z4D(T)-Psum}
\eeq
with the external potentials (\ref{phi4D_k}). 
We are thus led to the following conclusion. 

\begin{theorem}
As $R \to 0$ under the $R$-dependent transformation 
$\bst = \bst(\bsT,R)$ of the coupling constants 
defined by (\ref{t-T-rel}), 
$Z(\bst)$ converges to $Z_{\frD}(\bsT)$: 
\beq
  \lim_{R\to 0}Z(\bst(\bsT,R)) = Z_{\frD}(\bsT). 
\eeq
\end{theorem}

\begin{remark}
It is easy to see that $Z_{\frD}(X)$ and $Z_{\frD}(\bsT)$ 
are connected by the substitution 
\beq
  T_k = - \frac{\hbar^k}{kX^k} 
  \label{T-X-rel}
\eeq
as 
\beqnn
  Z_{\frD}(X) 
  = Z_{\frD}\left(-\frac{\hbar}{X}, -\frac{\hbar^2}{2X^2}, 
    \ldots, -\frac{\hbar^k}{kX^k},\dots\right). 
\eeqnn
This fact plays a role in the next section.  
\end{remark}

\begin{remark}
Although appears to be somewhat {\it ad hoc\/}, 
our prescription is essentially the only way to implement the 4D limit. 
This prescription is based on the natural relation 
(\ref{q^(ku)-u^k-rel}) that connects the exponential 
and polynomial functions in the external potentials. 
The $R$-dependent transformation (\ref{t-T-rel}) 
of the coupling constants is reminiscent of the method 
of various scaling limits in statistical mechanics 
and quantum field theory. 
\end{remark}

\section{Bilinear equations}

\subsection{Fay-type bilinear equations for KP hierarchy}

Let us recall the notion of Fay-type bilinear equations 
in the theory of the KP hierarchy \cite{AvM92,SS83,TT95}. 

Given a general tau function $\tau(\bst)$, one can consider 
an $N$-variate generalization of (\ref{tau(x)}):
\beqnn
  \tau([x_1]+\cdots+[x_N]) 
  = \langle 0|\Gamma_{+}(x_1,\ldots,x_N)|W\rangle. 
\eeqnn
Its product with the Vandermonde determinant 
\beqnn
  \Delta(x_1,\cdots,x_N) = \prod_{1\leq i<j\leq N}(x_i - x_j)
\eeqnn
is the $N$-point function of the fermion field $\psi^*(x^{-1})$ 
in the background state $|W\rangle$ \cite{JM83,MJD-book}.  

Actually, it is more convenient to leave $\bst$ as well.  
Let $\tau(\bst,x_1,\ldots,x_N)$ denote the function 
thus obtained: 
\begin{align}
  \tau(\bst,x_1,\ldots,x_N) 
  &= \tau(\bst + [x_1] + \cdots + [x_N]) \notag\\
  &= \langle 0|\Gamma_{+}(x_1,\ldots,x_N)
     \exp\left(\sum_{k=1}^\infty t_kJ_k\right)|W\rangle. 
  \label{tau(t,xx)}
\end{align}
By virtue of the aforementioned interpretation 
as the $N$-point function of a fermion field, the product 
\beqnn
  \xi(x_1,\ldots,x_N) = \Delta(x_1,\ldots,x_N)\tau(\bst,x_1,\ldots,x_N) 
\eeqnn
with the Vandermonde determinant satisfies the bilinear equations 
\beq
  \sum_{j=N}^{2N}(-1)^{j-N}\xi(x_1,\ldots,x_{N-1},x_j) 
    \xi(x_N,\ldots,\widehat{x_j},\ldots,x_{2N}) = 0, 
  \label{Fay-eq}
\eeq
where $\widehat{x_j}$ means removing $x_j$ 
from the list of variables therein. 
As pointed out by Sato and Sato \cite{SS83}, 
these equations are avatars of the Pl\"ucker relations 
among the Pl\"ucker coordinates of a Grassmann manifold. 

The simplest ($N = 2$) case 
\begin{align}
  &(x_1-x_2)(x_3-x_4)\tau(\bst+[x_1]+[x_2])\tau(\bst+[x_3]+[x_4]) \notag\\
  &\mbox{} 
   - (x_1-x_3)(x_2-x_4)\tau(\bst+[x_1]+[x_3])\tau(\bst+[x_2]+[x_4]) \notag\\
  &\mbox{} 
   + (x_1-x_4)(x_2-x_3)\tau(\bst+[x_1]+[x_4])\tau(\bst+[x_2]+[x_3]) = 0 
  \label{Fay4-eq}
\end{align}
of (\ref{Fay-eq}), referred to as a Fay-type bilinear equation, 
is known to play a particular role. 
Specialized to $x_4 = 0$, this equation turns 
into the so called Hirota--Miwa equation 
\begin{align}
  &(x_1-x_2)x_3\tau(\bst+[x_1]+[x_2])\tau(\bst+[x_3]) \notag\\
  &\mbox{} + (x_2-x_3)x_1\tau(\bst+[x_2]+[x_3])\tau(\bst+[x_1]) \notag\\
  &\mbox{} + (x_3-x_1)x_2\tau(\bst+[x_3]+[x_1])\tau(\bst+[x_2]) 
   = 0. 
  \label{HM-eq}
\end{align}
Moreover, dividing this equation by $x_3$ and letting $x_3 \to 0$ 
yield the differential Fay identity \cite{AvM92} 
\begin{align}
  &(x_1-x_2)\left(\tau(\bst+[x_1]+[x_2])\tau(\bst) 
            - \tau(\bst+[x_1])\tau(\bst+[x_2])\right) \notag\\
  &\mbox{} + x_1x_2\left(\tau(\bst+[x_1])\tau_{t_1}(\bst+[x_2])
              - \tau_{t_1}(\bst+[x_1])\tau(\bst+[x_2])\right) 
   = 0, 
   \label{diff-Fay-eq}
\end{align}
where $\tau_{t_1}(\bst)$ denotes the $t_1$-derivative of $\tau(\bst)$. 
It is known \cite{TT95} that the differential Fay identity 
characterizes a general tau function of the KP hierarchy 
in the following sense. 

\begin{theorem}
\label{KP-Fay-thm}
A function $\tau(\bst)$ of $\bst = (t_1,t_2,\ldots)$ 
is a tau function of the KP hierarchy if and only if 
it satisfies (\ref{diff-Fay-eq}). 
\end{theorem}

As a corollary, it turns out that each of (\ref{Fay4-eq}) 
and (\ref{HM-eq}), too, is a necessary and sufficient condition 
for a function $\tau(\bst)$ to be a KP tau function.  
This fact is a clue to the subsequent consideration.  

\begin{remark}
One can rewrite (\ref{Fay4-eq}) and (\ref{HM-eq}) 
to the equivalent forms 
\begin{align}
  &(x_1-x_2)(x_3-x_4)\tau(\bst-[x_1]-[x_2])\tau(\bst-[x_3]-[x_4]) \notag\\
  &\mbox{} 
   - (x_1-x_3)(x_2-x_4)\tau(\bst-[x_1]-[x_3])\tau(\bst-[x_2]-[x_4]) \notag\\
  &\mbox{} 
   + (x_1-x_4)(x_2-x_3)\tau(\bst-[x_1]-[x_4])\tau(\bst-[x_2]-[x_3]) = 0 
  \label{Fay4-eq2}
\end{align}
and 
\begin{align}
  &(x_1-x_2)x_3\tau(\bst-[x_1]-[x_2])\tau(\bst-[x_3]) \notag\\
  &\mbox{} + (x_2-x_3)x_1\tau(\bst-[x_2]-[x_3])\tau(\bst-[x_1]) \notag\\
  &\mbox{} + (x_3-x_1)x_2\tau(\bst-[x_3]-[x_1])\tau(\bst-[x_2]) 
   = 0. 
  \label{HM-eq2}
\end{align}
\end{remark}

\subsection{Bilinear equations in melting crystal model}

Let $Z(\bst,x_1,\ldots,x_N)$ denote the function 
\beq
  Z(\bst,x_1,\ldots,x_N) 
  = \sum_{\lambda\in\calP}s_\lambda(q^{-\rho})^2Q^{|\lambda|}
     e^{\phi(\bst,\lambda)}\prod_{j=1}^N\prod_{i=1}^\infty 
     \frac{1 - q^{\lambda_i-i+1}x_j}{1 - q^{-i+1}x_j}. 
  \label{Z(t,xx)-Psum}
\eeq
This function can be obtained from $Z(\bst)$ by shifting $\bst$ as 
\beqnn
  Z(\bst,x_1,\ldots,x_N) = Z(\bst - [x_1] - \cdots - [x_N]). 
\eeqnn
Note that the single-variate specialization $Z(x)$ of $Z(\bst)$ 
for the quantum spectral curve coincides with $Z(\bszero,q^{1/2}x)$. 
Since $Z(\bst)$ is a tau function of the KP hierarchy,  
the aforementioned bilinear equations imply, in particular, 
the three-term equation 
\begin{align}
  &(x_1-x_2)(x_3-x_4)Z(\bst,x_1,x_2)Z(\bst,x_3,x_4) \notag\\
  &\mbox{} - (x_1-x_3)(x_2-x_4)Z(\bst,x_1,x_3)Z(\bst,x_2,x_4) \notag\\
  &\mbox{} + (x_1-x_4)(x_2-x_3)Z(\bst,x_1,x_4)Z(\bst,x_2,x_3) = 0. 
  \label{Z-Fay4-eq}
\end{align}
as a consequence of (\ref{Fay4-eq2}). 

These bilinear equations turn out to survive the 4D limit. 
Let us set the parameters $q,Q$ and the coupling constants $\bst$ 
to the $R$-dependent form shown in (\ref{qQ-4Dlimit}) 
and (\ref{t-T-rel}), and transform the variables 
$x_1,\ldots,x_N$ to new variables $X_1,\ldots,X_N$ as 
\beq
  x_j = e^{RX_j},\quad j = 1,\ldots,N. 
  \label{x-X-rel2}
\eeq
Note that this transformation takes essentially the same form 
as the relation (\ref{x-X-rel}) between $x$ and $X$ 
in the 4D limit of the quantum spectral curve.  
As $R \to 0$ under these $R$-dependent transformations, 
$Z(\bst,x_1,\ldots,x_N)$  converges to a function of the form 
\begin{align}
  &Z_{\frD}(\bsT,X_1,\ldots,X_N) \notag\\
  &= \sum_{\lambda\in\calP}\left(\frac{\dim\lambda}{|\lambda|!}\right)^2 
     \left(\frac{\Lambda}{\hbar}\right)^{2|\lambda|}e^{\phi_{\frD}(\bsT,\lambda)} 
     \prod_{j=1}^N\prod_{i=1}^\infty
     \frac{X_j - (\lambda_i-i+1)\hbar}{X_j - (-i+1)\hbar}.
  \label{Z4D(T,XX)-Psum}
\end{align}
Since the differences $x_i - x_j$ in $\Delta(x_1,\ldots,x_N)$ 
behave as 
\beqnn
  x_i - x_j = R(X_i - X_j) + O(R^2), 
\eeqnn
the three-term bilinear equation (\ref{Z-Fay4-eq}), 
divided by $R^2$ before letting $R \to 0$, turns into the equation 
\begin{align}
  &(X_1 - X_2)(X_3 - X_4)Z_{\frD}(\bsT,X_1,X_2)Z_{\frD}(\bsT,X_3,X_4) \notag\\
  &\mbox{} - (X_1 - X_3)(X_2 - X_4)Z_{\frD}(\bsT,X_1,X_3)Z_{\frD}(\bsT,X_2,X_4) 
   \notag\\
  &\mbox{} + (X_1 - X_4)(X_2 - X_3)Z_{\frD}(\bsT,X_1,X_4)Z_{\frD}(\bsT,X_2,X_3) 
   = 0
  \label{Z4D-Fay4-eq}
\end{align}
for $Z_{\frD}(\bsT,X_i,X_j)$'s.  
The more general bilinear equations (\ref{Fay-eq}), 
too, have 4D counterparts. 

Let us note here that $Z_{\frD}(\bsT,X_1,\ldots,X_N)$ 
can be obtained from $Z_{\frD}(\bsT)$ by shifting $\bsT$ as 
\beqnn
  Z_{\frD}(\bsT,X_1,\ldots,X_N) 
  = Z_{\frD}(\bsT - [\hbar/X_1] - \cdots - [\hbar/X_N]) 
\eeqnn
just as $Z_{\frD}(X)$ and $Z_{\frD}(\bsT)$ are connected 
by the substitution shown in (\ref{T-X-rel}).
This is the same relation as $\tau(\bst,x_1,\ldots,x_N)$ 
is derived from $\tau(\bst)$ except that $x_j$'s 
are replaced by $h/X_j$'s. For more precise comparison 
with the bilinear equations for KP tau functions,  
one should rewrite (\ref{Z4D-Fay4-eq}) as 
\begin{align*}
  &(\hbar/X_1 - \hbar/X_2)(\hbar/X_3 - \hbar/X_4)
   Z_{\frD}(\bsT,X_1,X_2)Z_{\frD}(\bsT,X_3,X_4) \notag\\
  &\mbox{} - (\hbar/X_1 - \hbar/X_3)(\hbar/X_2 - \hbar/X_4)
             Z_{\frD}(\bsT,X_1,X_3)Z_{\frD}(\bsT,X_2,X_4) 
   \notag\\
  &\mbox{} + (\hbar/X_1 - \hbar/X_4)(\hbar/X_2 - \hbar/X_3)
             Z_{\frD}(\bsT,X_1,X_4)Z_{\frD}(\bsT,X_2,X_3) 
   = 0. 
\end{align*}
This equation literally corresponds to (\ref{Fay4-eq2}). 
According to Theorem \ref{KP-Fay-thm}, this is enough 
to deduce the following conclusion.  

\begin{theorem}
$Z_{\frD}(\bsT)$ is a tau function of the KP hierarchy. 
\end{theorem}

\subsection{Extension to Toda hierarchy}

$Z_{\frD}(\bsT)$ has a fermionic expression, 
analogous to (\ref{Z(t)=<..e^H..>}), of the form 
\beq
  Z_{\frD}(\bsT)
  = \langle 0|e^{J_1}(\Lambda/\hbar)^{2L_0}
    e^{H_{\frD}(\bsT)}e^{J_{-1}}|0\rangle, 
  \label{Z4D(t)=<..e^H..>}
\eeq
where 
\beqnn
  H_{\frD}(\bsT) = \sum_{k=1}^\infty T_kH^{\frD}_k,\quad 
  H^{\frD}_k = \sum_{n\in\ZZ}n^k{:}\psi_{-n}\psi^*_n{:}. 
\eeqnn
It is natural to extend this function to a set 
of functions $Z_{\frD}(\bsT,s)$, $s \in \ZZ$, as 
\beqnn
  Z_{\frD}(\bsT,s) 
  = \langle s|e^{J_1}(\Lambda/\hbar)^{2L_0}
    e^{H_{\frD}(\bsT)}e^{J_{-1}}|s\rangle, 
\eeqnn
where $|s\rangle$ and $\langle s|$ are the ground states 
of the charge-$s$ sector of the Fock spaces.  
These partition functions have the combinatorial expression 
\beq
  Z_{\frD}(\bsT,s)
  = \sum_{\lambda\in\calP}\left(\frac{\dim\lambda}{|\lambda|!}\right)^2 
      \left(\frac{\Lambda}{\hbar}\right)^{2|\lambda|+s(s+1)}
      e^{\phi_{\frD}(\bsT,s,\lambda)}, 
  \label{Z4D(t,s)-Psum}
\eeq
where 
\begin{gather*}
  \phi_{\frD}(\bsT,s,\lambda) = \sum_{k=1}^\infty T_k\phi^{\frD}_k(s,\lambda),\\
  \phi^{\frD}_k(s,\lambda) 
  = \sum_{i=1}^\infty\left((\lambda_i-i+1+s)^k - (-i+1+s)^k\right) 
     + \text{correction terms}. 
\end{gather*}
The external potentials are obtained by rearrangement 
of terms in the formal expression 
\beqnn
  \phi^{\frD}_k(s,\lambda) 
  = \sum_{i=1}^\infty(\lambda_i-i+1+s)^k - \sum_{i=1}^\infty(-i+1)^k, 
\eeqnn
hence 
\begin{align}
  \text{correction terms} 
  &= \sum_{i=1}^\infty(-i+1+s)^k - \sum_{i=1}^\infty(-i+1)^k \notag\\
  &= \begin{cases}
    1^k + \cdots + s^k& \text{for $s > 0$},\\
    0 & \text{for $s = 0$},\\
    - (-1)^k - \cdots - (s+1)^k& \text{for $s < 0$}. 
    \end{cases}
    \label{4Dphi-corr}
\end{align}
$Z_{\frD}(\bsT,s)$ is known as a generating function 
of all-genus Gromov-Witten invariants of $\CC\PP^1$ \cite{OP02}, 
and proven to be a tau function of the 1D Toda hierarchy 
\cite{DZ04,Getzler01,Milanov06}.  

The foregoing approach to the integrable structure 
of $Z_{\frD}(\bsT)$ can be carried over to $Z_{\frD}(\bsT,s)$.  
Namely, one can use a Toda version \cite{Takasaki07,Teo06} 
of Fay-type bilinear equations to prove the following 
\footnote{A full generating function of all-genus 
Gromov-Witten invariants of $\CC\PP^1$ \cite{OP02} 
depends on another set of variables $s_1,s_2,\ldots$ 
(the decedents of $s$) as well.  
These variables are now set to $0$, 
though the Toda conjecture is proven in the presence 
of these variables and the associated time evolutions 
(the extended Toda hierarchy \cite{CDZ04}).  
In this sense, our proof is incomplete as an alternative proof 
of the Toda conjecture.}: 

\begin{theorem}
\label{Z4D-Toda-thm}
$Z_{\frD}(\bsT,s)$ is a tau function of the 1D Toda hierarchy. 
\end{theorem}

The proof is far more complicated than the case of $Z_{\frD}(\bsT)$.  
We present its detail in Appendix.

\section{Conclusion}

Primary motivation of this work was to understand the result 
of Dunin-Barkowski et al. \cite{DBMNPS13} in the language of 
the quantum spectral curve of the melting crystal model. 
In the course of solving this problem, we have found 
how to achieve the 4D limit of the deformed partition 
function $Z(\bst)$ itself. As a byproduct, this prescription 
for 4D limit has turned out to transfer Fay-type bilinear equations 
from $Z(\bst)$ to its 4D limit $Z_{\frD}(\bsT)$.  

It will be better to summarize these results from two aspects, 
namely, quantum curves and bilinear equations: 

\begin{itemize}
\item[1.] {\it Quantum curves\/}: 
One can derive a quantum spectral curve of the melting crystal model 
by the method of our work on quantum mirror curves 
in topological string theory \cite{TN16}. This quantum curve 
is formulated as the $q$-difference equation (\ref{(A-1)Z(x)=0}) 
for the single-variate specialization $Z(x)$ of $Z(\bst)$.  
Its 4D limit is achieved by transforming the variable $x$ 
to a new variable $X$ as shown in (\ref{x-X-rel}) 
and letting $R \to 0$.  (\ref{(A-1)Z(x)=0}) thereby 
turns into the difference equation (\ref{Z4D(X)-diffeq}) 
for the 4D version $Z_{\frD}(X)$ of $Z(x)$.  
(\ref{Z4D(X)-diffeq}) can be further converted 
to the quantum spectral curve (\ref{CP1qsc-eq}) 
of Gromov-Witten theory of $\CC\PP^1$. 
(\ref{Z4D(X)-diffeq}) and (\ref{CP1qsc-eq}) are derived 
by Dunin-Barkowski et al. \cite{DBMNPS13} by 
genuinely combinatorial computations. 
Our approach highlights a role of the KP hierarchy 
that underlies these quantum curves.  
\item[2.] {\it Bilinear equations\/}: 
According to our previous work on the melting crystal model \cite{NT07}, 
$Z(\bst)$ is a tau function of the KP hierarchy.  
As $R \to 0$ under the $R$-dependent transformation (\ref{t-T-rel}) 
of the coupling constants, $Z(\bst)$ converges 
to the 4D version $Z_{\frD}(\bsT)$.  In this limit, 
the three-term bilinear equation (\ref{Z-Fay4-eq}) for $Z(\bst)$ 
turns into its counterpart (\ref{Z4D-Fay4-eq}) for $Z_{\frD}(\bsT)$. 
This implies that $Z_{\frD}(\bsT)$, too, is a tau function 
of the KP hierarchy.  Thus we have obtained a new approach 
to the integrable structure in Okounkov and Pandharipande's 
generating function of all-genus Gromov-Witten invariants 
of $\CC\PP^1$ \cite{OP02}. 
\end{itemize}

As explained in Appendix, the method of 4D limit 
for bilinear equations can be extend to the $s$-deformed 
partition functions $Z(\bst,s)$ and $Z_{\frD}(\bsT,s)$. 
This is a yet another proof of the fact \cite{DZ04,Getzler01,Milanov06} 
that $Z_{\frD}(\bsT,s)$ is a tau function of the 1D Toda hierarchy.  

Let us stress that the integrable structure of $Z_{\frD}(\bsT)$ 
still remains to be fully elucidated.  
Its 5D (or K-theoretic) lift $Z(\bst)$ has a fermionic expression, 
such as (\ref{Z(t)=<..g_3..>}), that shows manifestly 
that $Z(\bst)$ is a tau function of the KP hierarchy.  
Moreover, since the generating operator 
in the fermionic expression is rather simple, 
one can even find the associated generating operator $G$ 
in $V = \CC((x))$ explicitly.  
In contrast, no similar fermionic expression of $Z_{\frD}(\bsT)$ 
is currently known.  The preliminary fermionic expression 
(\ref{Z4D(t)=<..e^H..>}) of $Z_{\frD}(\bsT)$ cannot be converted 
to such a form by the method of our previous work \cite{NT07}.  
The limiting procedure from $Z(\bst)$ is a way to overcome 
this difficulty, but this will not be a final answer. 

An alternative approach will be the route 
from the equivariant Gromov-Witten theory. 
A generating function of the equivariant Gromov-Witten invariants 
of $\CC\PP^1$ is known to be a tau function 
of the 2D Toda hierarchy \cite{OP02b}.  
This generating function should reproduce $Z_{\frD}(\bsT,s)$ 
in the limit as the equivariant parameter tends to $0$.  
We hope to address this issue elsewhere.

\subsection*{Acknowledgements}

The author is grateful to Toshio Nakatsu for valuable comments. 
This work is partly supported by the JSPS Kakenhi Grant 
JP25400111, JP15K04912 and JP18K03350.

\appendix

\section{Derivation of Toda hierarchy}

We prove Theorem \ref{Z4D-Toda-thm} in this appendix.  
This proof is based on a 1D Toda version of Theorem \ref{KP-Fay-thm}. 
This theorem characterizes all tau functions of the 1D Toda hierarchy 
by two bilinear equations.  The 5D partition function itself 
is not a genuine tau function, and satisfies slightly modified 
bilinear equations.  In the 4D limit, these equations turn 
into bilinear equations for the 4D partition function.  
These equations agree with the bilinear equations 
for 1D Toda tau functions.  Thus we can conclude that 
the 4D partition function is a tau function of the 1D Toda hierarchy.  

\subsection{Toda tau function in 5D partition function}

Let us consider the $s$-deformed 5D partition function $Z(\bst,s)$ 
written in the fermionic form 
\beq
  Z(\bst,s) = \langle s|\Gamma_{+}(q^{-\rho})Q^{L_0}
            e^{H(\bst)}\Gamma_{-}(q^{-\rho})|s\rangle. 
\eeq
Just as the fermionic expressions (\ref{Z(t)=<..e^H..>}) 
of $Z(\bst)$ is converted to (\ref{Z(t)=<..g_1..>}), 
one can rewrite $Z(\bst,s)$ \cite{NT07} as 
\beq
  Z(\bst,s) = q^{-(4s^3-s)/12}\exp\left(\sum_{k=1}^\infty\frac{q^kt_k}{1-q^k}\right) 
    \langle s|\exp\left(\sum_{k=1}^\infty(-1)^kq^{k/2}t_kJ_k\right)
    g_1|s\rangle.
\eeq
Moreover, one can remove the multipliers $(-1)^kq^{k/2}$ 
of $t_k$'s as 
\begin{align}
  Z(\bst,s) 
  &= q^{-(4s^3-s)/12}q^{-s(s+1)/2} 
     \exp\left(\sum_{k=1}^\infty\frac{q^kt_k}{1-q^k}\right) \notag\\
  &\quad\mbox{}\times
    \langle s|\exp\left(\sum_{k=1}^\infty t_kJ_k\right)
    (-q^{1/2})^{L_0}g_1(-q^{1/2})^{L_0}|s\rangle. 
   \label{Z(t,s)=<..g_1..>}
\end{align}
Note that the $s$-dependent prefactors of the vevs 
originate in the relations 
\beqnn
\begin{gathered}
  \langle s|q^{-K/2} = q^{-(4s^3-s)/24}\langle s|,\quad 
  q^{-K/2}|s\rangle = q^{-(4s^3-s)/24}|s\rangle,\\
  \langle s|(-q^{1/2})^{-L_0} = (-q^{1/2})^{-s(s+1)/2}\langle s|,\quad 
  (-q^{1/2})^{-L_0}|s\rangle = (-q^{1/2})^{-s(s+1)/2}|s\rangle. 
\end{gathered}
\eeqnn

Let $\tau(\bst,s)$ denote the vev in (\ref{Z(t,s)=<..g_1..>}): 
\beq
  \tau(\bst,s) = \langle s|\exp\left(\sum_{k=1}^\infty t_kJ_k\right)
    (-q^{1/2})^{L_0}g_1(-q^{1/2})^{L_0}|s\rangle. 
  \label{tau(t,s)}
\eeq
As a consequence of (\ref{g_1-symmetry}), one can move 
the exponential operator of $J_k$'s to the right as 
\beq
  \tau(\bst,s) = \langle s|(-q^{1/2})^{L_0}g_1(-q^{1/2})^{L_0}
  \exp\left(\sum_{k=1}^\infty t_kJ_k\right)|s\rangle. 
\eeq
This means that $\tau(\bst,s)$ is a tau function 
of the 1D Toda hierarchy \cite{UT84}. 
The product of this function with the exponential function 
of $t_k$'s in (\ref{Z(t,s)=<..g_1..>}), too, is a tau function 
of the 1D Toda hierarchy.  This is, however, not the case 
for further multiplication by $q^{-(4s^3-s)/12}q^{-s(s+1)/2}$. 
As we show below, these two $s$-dependent prefactors 
modify the bilinear equations.

\subsection{Bilinear equations for 5D partition function}

The 1D Toda hierarchy can be characterized 
by the two bilinear equations shown below.  
This is a consequence of the characterization 
of the 2D Toda hierarchy by three Fay-type identities 
\cite{Takasaki07,Teo06}. 

\begin{theorem}
A function $\tau(\bst,s)$ of $\bst = (t_1,t_2,\ldots)$ and $s$ 
is a tau functions of the 1D Toda hierarchy if and only if 
it satisfies the following bilinear equations: 
\begin{align}
  (y-x)\tau(\bst,s+1)\tau(\bst-[x]-[y],s) 
  - y\tau(\bst-[x],s+1)\tau(\bst-[y],s)& \notag\\
  \mbox{} + x\tau(\bst-[y],s+1)\tau(\bst-[x],s) &= 0, 
    \label{Toda-Fay1}\\
  \tau(\bst-[x],s)\tau(\bst-[y],s) 
  - \tau(\bst,s)\tau(\bst-[x]-[y],s)& \notag\\
  \mbox{}+ xy\tau(\bst,s+1)\tau(\bst-[x]-[y],s-1) &= 0. 
    \label{Toda-Fay2}
\end{align}
\end{theorem}

The tau function (\ref{tau(t,s)}) and its product with 
the exponential function of $t_k$'s in (\ref{Z(t,s)=<..g_1..>}) 
satisfy these equations.  We convert these equations 
to bilinear equations for $Z(\bst,s)$.  It is convenient 
to separate some simple factors from $Z(\bst,s)$ as  
\beq
  Z(\bst,s) = Q^{s(s+1)/2}e^{\phi(\bst,s,\emptyset)}\tilde{Z}(\bst,s),
  \label{Z-Zt-rel}
\eeq
where
\beqnn
\begin{gathered}
  \tilde{Z}(\bst,s) 
    = \sum_{\lambda\in\calP}s_\lambda(q^{-\rho})^2Q^{|\lambda|}
      e^{\tilde{\phi}(\bst,s,\lambda)},\\
  \tilde{\phi}(\bst,s,\lambda) 
    = \sum_{k=1}^\infty t_k\tilde{\phi}_k(s,\lambda) 
    = \phi(\bst,s,\lambda) - \phi(\bst,s,\emptyset). 
\end{gathered}
\eeqnn
Note that $\phi(\bst,s,\emptyset)$ consists of 
the ``correction terms'' (\ref{phi-corr}): 
\beqnn
\begin{gathered}
  \phi(\bst,s,\emptyset) = \sum_{k=1}^\infty t_k\phi_k(s,\emptyset),\\
  \phi_k(s,\emptyset) 
    = \sum_{i=1}^\infty q^{k(-i+1+s)} - \sum_{i=1}^\infty q^{k(-i+1)} 
    = \frac{1-q^{ks}}{1-q^k}q^k. 
\end{gathered}
\eeqnn
The relation between the tau function and 
the 5D partition function thereby takes such a form as 
\beq
  \tau(\bst,s) = \left(-\sum_{k=1}^\infty\frac{q^kt_k}{1-q^k}\right) 
   q^{(4s^3-s)/12}(qQ)^{s(s+1)/2}e^{\phi(\bst,s,\emptyset)}\tilde{Z}(\bst,s). 
\eeq

Plugging this expression into (\ref{Toda-Fay1}) and (\ref{Toda-Fay2}),  
we obtain bilinear equations for $\tilde{Z}(\bst,s)$.  
The prefactors $q^{(4s^3-s)/12}$, $(qQ)^{s(s+1)/12}$ 
and $e^{\phi(\bst,s,\emptyset)}$ of $\tilde{Z}(\bst,s)$ yield 
extra factors in the bilinear equations, e.g., 
\beqnn
\begin{aligned}
  \frac{e^{\phi(\bst-[x],s+1,\emptyset)}e^{\phi(\bst-[y],s,\emptyset)}}
        {e^{\phi(\bst,s+1,\emptyset)}e^{\phi(\bst-[x]-[y],s,\emptyset)}}
  &= \exp\left(- \sum_{k=1}^\infty\frac{x^k}{k}
       (\phi_k(s+1,\emptyset) - \phi_k(s,\emptyset))\right)\\
  &= \exp\left(- \sum_{k=1}^\infty\frac{x^k}{k}q^{k(s+1)}\right)\\
  &= 1 - q^{s+1}x. 
\end{aligned}
\eeqnn
In the same sense, we have the following extra factors: 
\beqnn
  \frac{e^{\phi(\bst-[y],s+1,\emptyset)}e^{\phi(\bst-[x],s,\emptyset)}}
        {e^{\phi(\bst,s+1,\emptyset)}e^{\phi(\bst-[x]-[y],s,\emptyset)}}
  = 1 - q^{s+1}y, 
\eeqnn
\beqnn
  \frac{q^{(4(s+1)^3-(s+1))/12}q^{(4(s-1)^3-(s-1))/12}}{q^{(4s^3-s)/6}} = q^{2s},
\eeqnn
\beqnn
  \frac{(qQ)^{(s+1)(s+2)/2}(qQ)^{(s-1)s/2}}{(qQ)^{s(s+1)}} = qQ,
\eeqnn
\beqnn
  \frac{e^{\phi(\bst,s+1,\emptyset)}e^{\phi(\bst-[x]-[y],s-1,\emptyset)}}
        {e^{\phi(\bst-[x],s,\emptyset)}e^{\phi(\bst-[y],s,\emptyset)}}
  = \frac{\exp\left(\sum_{k=1}^\infty(q^{k(s+1)} - q^{ks})t_k\right)}
         {(1 - q^sx)(1 - q^sy)}. 
\eeqnn
Thus the bilinear equations for $\tilde{Z}(\bst,s)$ take 
the following form: 
\begin{align}
  (y-x)\tilde{Z}(\bst,s+1)\tilde{Z}(\bst-[x]-[y],s)& \notag\\
  \mbox{} - (1-q^{s+1}x)y\tilde{Z}(\bst-[x],s+1)\tilde{Z}(\bst-[y],s)& \notag\\
  \mbox{} + (1-q^{s+1}y)x\tilde{Z}(\bst-[y],s+1)\tilde{Z}(\bst-[x],s) &= 0, 
    \label{Zt-bilineq1}\\
  \tilde{Z}(\bst-[x],s)\tilde{Z}(\bst-[y],s) 
  - \tilde{Z}(\bst,s)\tilde{Z}(\bst-[x]-[y],s)& \notag\\
  \mbox{} + \frac{q^{2s}qQ}{(1-q^sx)(1-q^sy)}
            \exp\left(\sum_{k=1}^\infty(q^{k(s+1)}-q^{ks})t_k\right)& \notag\\
  \mbox{} \times xy\tilde{Z}(\bst,s+1)\tilde{Z}(\bst-[x]-[y],s-1) &= 0. 
    \label{Zt-bilineq2}
\end{align}

\subsection{4D limit of 5D partition function}

The prescription of 4D limit for $Z(\bst)$ can be carried over 
to $\tilde{Z}(\bst,s)$.  Note that the external potentials 
$\tilde{\phi}_k(s,\lambda)$ in $\tilde{Z}(\bst,s)$ 
are $s$-dependent analogues of $\phi_k(\lambda)$'s: 
\beqnn
  \tilde{\phi}_k(s,\lambda) 
    = \sum_{i=1}^\infty\left(q^{k(\lambda_i-i+1+s)} - q^{k(-i+1+s)}\right). 
\eeqnn
Consequently, as $R \to 0$ under the $R$-dependent setting 
(\ref{qQ-4Dlimit}) of $q$ and $Q$, we have the relation 
\beq
  \sum_{j=1}^k(-1)^{k-j}\binom{k}{j}\tilde{\phi}_j(s,\lambda) 
  = \tilde{\phi}^{\frD}_k(s,\lambda)(-R\hbar)^k + O(R^{k+1}), 
\eeq
where 
\beqnn
  \tilde{\phi}^{\frD}_k(s,\lambda) 
  = \sum_{i=1}^\infty\left((\lambda_i-i+1+s)^k - (-i+1+s)^k\right), 
\eeqnn
hence a meaningful limit 
\beq
  \lim_{R\to 0}\tilde{Z}(\bst(\bsT,R),s) = \tilde{Z}_{\frD}(\bsT,s) 
\eeq
under the same $R$-dependent parametrization (\ref{t-T-rel}) 
of the coupling constants as used for the 4D limit of $Z(\bst)$.  
$\tilde{Z}_{\frD}(\bsT,s)$ is an $s$-dependent analogue 
of $Z_{\frD}(\bsT)$ of the form 
\beq
  \tilde{Z}_{\frD}(\bsT,s)
  = \sum_{\lambda\in\calP}\left(\frac{\dim\lambda}{|\lambda|!}\right)^2 
      \left(\frac{\Lambda}{\hbar}\right)^{2|\lambda|}
      e^{\tilde{\phi}_{\frD}(\bsT,s,\lambda)} 
\eeq
with the external potential 
\beqnn
  \tilde{\phi}_{\frD}(\bsT,s,\lambda) 
  = \sum_{k=1}^\infty T_k\tilde{\phi}^{\frD}_k(s,\lambda). 
\eeqnn

Let us recall that the correction term (\ref{4Dphi-corr}) 
of $\phi^{\frD}_k(s,\lambda)$ is nothing but $\phi^{\frD}_k(s,\emptyset)$. 
In other words, we have the relation 
\beqnn
  \tilde{\phi}_{\frD}(\bsT,s,\lambda) 
  = \phi_{\frD}(\bsT,s,\lambda) - \phi_{\frD}(\bsT,s,\emptyset). 
\eeqnn
Thus $\tilde{Z}_{\frD}(\bsT,s)$ is related to 
the full 4D partition function (\ref{Z4D(t,s)-Psum}) as 
\beq
  Z_{\frD}(\bsT,s) 
  = (\Lambda^2/\hbar^2)^{s(s+1)/2}e^{\phi_{\frD}(\bsT,s,\emptyset)}\tilde{Z}_{\frD}(\bsT,s) 
  \label{Z4D-Zt4D-rel} 
\eeq
just like the relation (\ref{Z-Zt-rel}) 
between $Z(\bst,s)$ and $\tilde{Z}(\bst,s)$.

\subsection{Bilinear equations for 4D partition function}

We can derive bilinear equations for $\tilde{Z}_{\frD}(\bsT,s)$ 
from (\ref{Zt-bilineq1}) and (\ref{Zt-bilineq2}) 
by setting 
\beqnn
  x = e^{RX},\quad y = e^{RY} 
\eeqnn
just like (\ref{x-X-rel}) and (\ref{x-X-rel2}), 
and letting $R \to 0$ under the $R$-dependent parametrization 
(\ref{qQ-4Dlimit}), (\ref{t-T-rel}) of $q,Q$ and $\bst$.  

Deriving a 4D counterpart of (\ref{Zt-bilineq1}) 
is straightforward. In the limit as $R \to 0$, 
the shifted partition functions in (\ref{Zt-bilineq1}) 
turn into shifted partition functions of $\tilde{Z}_{\frD}(\bsT,s)$ as 
\beqnn
\begin{gathered}
  \tilde{Z}(\bst - [x], s) \to \tilde{Z}_{\frD}(\bsT - [\hbar/X], s),\\
  \tilde{Z}(\bst - [y], s) \to \tilde{Z}_{\frD}(\bsT - [\hbar/Y], s),\\
  \tilde{Z}(\bst - [x] - [y], s) 
    \to \tilde{Z}_{\frD}(\bsT - [\hbar/X] - [\hbar/Y], s) 
\end{gathered}
\eeqnn
just as shown in the derivation of (\ref{Z4D(T,XX)-Psum}). 
Other simple factors in (\ref{Zt-bilineq1}) behave as 
\beqnn
\begin{gathered}
  x = 1 + O(R),\quad y = 1 + O(R),\\
  y - x = R(Y - X) + O(R),\\
  1 - q^{s+1}x = (\hbar(s+1) - X)R + O(R),\\
  1 - q^{s+1}y = (\hbar(s+1) - Y)R + O(R). 
\end{gathered}
\eeqnn
Thus we obtain the following bilinear equation: 
\begin{align}
    (Y-X)\tilde{Z}_{\frD}(\bst,s+1)
      \tilde{Z}_{\frD}(\bst-[\hbar/X]-[\hbar/Y],s)& \notag\\
  \mbox{} - (Y - \hbar(s+1))\tilde{Z}_{\frD}(\bst-[\hbar/X],s+1)
            \tilde{Z}_{\frD}(\bst-[\hbar/Y],s)& \notag\\
  \mbox{} + (X - \hbar(s+1))\tilde{Z}_{\frD}(\bst-[\hbar/Y],s+1)
            \tilde{Z}_{\frD}(\bst-[\hbar/X],s) &= 0. 
  \label{Zt4D-bilineq3}
\end{align}

To derive a 4D counterpart of (\ref{Zt-bilineq2}), 
we have to cope with the strange exponential factor therein. 
By its origin, this factor is related to $\phi(\bst,s,\emptyset)$ 
and $\phi(\bst,s\pm 1,\emptyset)$ as 
\beqnn
  \exp\left(\sum_{k=1}^\infty(q^{k(s+1)} - q^{ks})t_k\right) 
  = \frac{e^{\phi(\bst,s+1,\emptyset)}e^{\phi(\bst,s-1,\emptyset)}}
        {e^{2\phi(\bst,s,\emptyset)}}. 
\eeqnn
Consequently, upon substituting $q = e^{-R\hbar}$ and $t_k = t_k(\bsT,R)$, 
we can take the limit as $R \to 0$: 
\beqnn
  \lim_{R\to 0}\exp\left(\sum_{k=1}^\infty(q^{k(s+1)} - q^{ks})t_k\right) 
  = \frac{e^{\phi_{\frD}(\bsT,s+1,\emptyset)}e^{\phi_{\frD}(\bsT,s-1,\emptyset)}}
        {e^{2\phi_{\frD}(\bsT,s,\emptyset)}}. 
\eeqnn
Thus (\ref{Zt-bilineq2}) turns into the following bilinear equation: 
\begin{align}
  \tilde{Z}_{\frD}(\bsT-[\hbar/X],s)\tilde{Z}_{\frD}(\bsT-[\hbar/Y],s)&\notag\\
  \mbox{} - \tilde{Z}_{\frD}(\bsT,s)
            \tilde{Z}_{\frD}(\bsT-[\hbar/X]-[\hbar/Y],s)&\notag\\
  \mbox{} + \frac{\Lambda^2}{(X-\hbar s)(Y-\hbar s)}
            \frac{e^{\phi_{\frD}(\bsT,s+1,\emptyset)}e^{\phi_{\frD}(\bsT,s-1,\emptyset)}}
            {e^{2\phi_{\frD}(\bsT,s,\emptyset)}}&\notag\\
  \mbox{} \times\tilde{Z}_{\frD}(\bsT,s+1)
                \tilde{Z}_{\frD}(\bsT-[\hbar/X]-[\hbar/Y],s-1) &= 0. 
    \label{Zt4D-bilineq4}
\end{align}

Since $\tilde{Z}_{\frD}(\bsT,s)$ and $Z_{\frD}(\bsT,s)$ 
are linearly related as shown in (\ref{Z4D-Zt4D-rel}), 
we can translate the foregoing bilinear equations 
(\ref{Zt4D-bilineq3}) and (\ref{Zt4D-bilineq4}) 
for $\tilde{Z}_{\frD}(\bsT,s)$ to equations for $Z_{\frD}(\bsT,s)$ 
in much the same way as the derivation 
of (\ref{Zt-bilineq1}) and (\ref{Zt-bilineq2}). 
Thus we obtain the bilinear equations 
\begin{align}
  (Y-X)Z_{\frD}(\bsT,s+1)Z_{\frD}(\bsT-[\hbar/X]-[\hbar/Y], s)&\notag\\
  \mbox{} - XZ_{\frD}(\bsT-[\hbar/X],s+1)Z_{\frD}(\bsT-[\hbar/Y],s)&\notag\\
  \mbox{} + YZ_{\frD}(\bsT-[\hbar/Y],s+1)Z_{\frD}(\bsT-[\hbar/X],s)& = 0
\end{align}
and
\begin{align}
  Z_{\frD}(\bsT-[\hbar/X],s)Z_{\frD}(\bsT-[\hbar/Y],s)&\notag\\
  \mbox{} - Z_{\frD}(\bsT,s)Z_{\frD}(\bsT-[\hbar/X]-[\hbar/Y],s)&\notag\\
  \mbox{} + (\hbar^2/XY)Z_{\frD}(\bsT,s+1)
            Z_{\frD}(\bsT-[\hbar/X]-[\hbar/Y],s-1) &= 0 
\end{align}
for $Z_{\frD}(\bsT,s)$.  By replacing $\hbar/X \to x$ 
and $\hbar/Y \to y$, these equations turn into a form 
that is identical to (\ref{Toda-Fay1}) and (\ref{Toda-Fay2}). 
This implies that $Z_{\frD}(\bsT,s)$ is a tau function 
of the 1D Toda hierarchy.

\end{document}